\documentclass[journal]{IEEEtran}
\usepackage[utf8]{inputenc}
\usepackage[english]{babel}
\usepackage{booktabs}
\usepackage{multirow}
\usepackage{amssymb}
\usepackage[table,xcdraw,dvipsnames]{xcolor}
\usepackage{cite}
\usepackage{float}
\usepackage{graphicx}

\usepackage{algorithm}
\usepackage{algorithmic}
\usepackage[cmex10]{amsmath}
\newcommand{\lora}{\textsc{LoRa}}
\newcommand{\lorawan}{\textsc{LoRaWAN}}
\newcommand{\sigfox}{\textsc{SigFox}}

\begin{document}

\title{Hybrid Coded Replication in \lora{} Networks}

\author{Jean~Michel~de~Souza~Sant'Ana, Arliones~Hoeller, Richard~Demo~Souza, \\Samuel~Montejo-Sánchez,  Hirley Alves, and Mario~de~Noronha~Neto

\thanks{Jean Michel de Sousa Sant'Ana, Arliones Hoeller and Hirley Alves are with Centre for Wireless Communications, University of Oulu, Oulu, Finland (\{Jean.DeSouzaSantana, Hirley.Alves\}@oulu.fi).}
\thanks{Arliones Hoeller and Mario de Noronha Neto are with the Department of Telecommunications, Federal Institute for Education, Science and Technology of Santa Catarina, São José, Brazil (\{arliones.hoeller, noronha\}@ifsc.edu.br)}
\thanks{Richard Demo Souza and Arliones Hoeller are with the Department of Electrical and Electronics Engineering, Federal University of Santa Catarina, Florianóplis, Brazil (richard.demo@ufsc.br)}
\thanks{Samuel Montejo-S\'anchez  is with Programa Institucional de Fomento a la I+D+i, Universidad Tecnol\'ogica Metropolitana, Santiago, Chile. (smontejo@utem.cl)}
\thanks{This work has been partially supported in Brazil by CNPq, FUMDES-UNIEDU, Print CAPES-UFSC Project Automation 4.0, and INESC Brazil Project F-LOCO (Energisa/ANEEL PD-00405-1804/2018); in Finland by Academy of Finland (Aka) 6Genesis Flagship (Grant 318927), EE-IoT (Grant 319008), Aka Prof (Grant 307492), and FIREMAN (Grant 326301); and in Chile by FONDECYT Postdoctoral (Grant 3170021).}
\thanks{\textcopyright~2020 IEEE. Personal use of this material is permitted. Permission from IEEE must be obtained for all other uses, in any current or future media,
including reprinting/republishing this material for advertising or promotional
purposes, creating new collective works, for resale or redistribution to servers or lists, or reuse of any copyrighted component of this work in other works.}
}

\maketitle

\begin{abstract}
Low Power Wide Area Networks (LPWAN) are wireless connectivity solutions for Internet-of-Things (IoT) applications, including industrial automation.
Among the several LPWAN technologies, \lorawan{} has been extensively addressed by the research community and the industry.
However, the reliability and scalability of \lorawan{} are still uncertain.
One of the techniques to increase the reliability of \lorawan{} is message replication, which exploits time diversity.
This paper proposes a novel hybrid coded message replication scheme that interleaves simple repetition and a recently proposed coded replication method.
We analyze the optimization of the proposed scheme under minimum reliability requirements and show that it enhances the network performance without requiring additional transmit power compared to the competing replication techniques.
\end{abstract}

\section{Introduction}

Many Internet-of-Things (IoT) applications~\cite{Centenaro:WC:2016,Tome:TIE:2018} require wireless coverage of high-density areas using low power devices with long battery lifetime. Such demands make short-range technologies, like Bluetooth and ZigBee, as well as established cellular networks, like GSM and LTE, unfeasible solutions for different reasons, including cost, energy consumption, and limited coverage. Low Power Wide Area Network (LPWAN) technologies like \lorawan{}, \sigfox{}, and NB-IoT~\cite{Raza:CST:2017}, are natural candidates for such applications due to their long-range, large-capacity, and power efficiency~\cite{Lauridsen:VTCS:2017}.

Moreover, very low bit rate wide area networks (VLBR WAN), such as those built based on LPWAN technologies, are viable wireless solutions for several Industrial IoT applications demanding the coverage of several devices, with latency and reliability constraints that are somewhat milder for the industrial scenario. Among these applications, we may cite~\cite{Candell:IEM:2018}: monitoring and supervisory control of flow-based systems; factory monitoring of job-based systems; tracking of personnel, tools, and materials; as well as machinery health monitoring. 

Among the LPWAN technologies, \lorawan{} network protocols~\cite{LoRaAlliance:2017}, based on \lora{} physical layer technology~\cite{Semtech:2015}, is attracting much attention not only from the academy but also from industry~\cite{LoRaWANVerticalMarkets}. However, the ultimate performance limits of \lorawan{} are still unknown, while many open questions remain on how to increase or even to guarantee the reliability demanded by typical industrial applications.
This paper focuses on the capacity of \lorawan{} in terms of the number of devices in coverage, while meeting a given target reliability constraint, and when using message replication techniques to exploit time diversity.
We build on the repetition replication scheme investigated in~\cite{Hoeller:ACS:2018} and on a coded replication scheme recently proposed in~\cite{Montejo-Sanchez:WCL:2019}, to introduce a novel hybrid replication approach optimized for \lorawan, which can increase the network performance with respect to~\cite{Hoeller:ACS:2018} and~\cite{Montejo-Sanchez:WCL:2019}.

\subsection{Related Work}

Georgiou and Raza~\cite{Georgiou:WCL:2017} present a mathematical model taking into account connection and collision probabilities to analyze large-scale \lora{} networks. They show that collisions are the most limiting factor.
Mahmood \textit{et al.}~\cite{Mahmood:TII:2019} present a similar model that considers accumulated interference rather than only the highest interference as in~\cite{Georgiou:WCL:2017}, besides inter-SF interference. They show that considering the strongest colliding node is a good approximation only for low-density scenarios.  Mikhaylov \textit{et al.}~\cite{Mikhaylov:EWC:2016} present a scalability analysis of \lorawan{}, optimizing the number of nodes per SF.

There are several studies about the performance of LPWAN technologies   in industrial scenarios. Sommer \textit{et al.}~\cite{Sommer:ICII:2018} discuss and evaluate several LPWAN protocols for indoor industrial scenarios.
They conclude that \lorawan{} provides the best link budget and reliability.
Rizzi \textit{et al.}~\cite{Rizzi:WFCS:2017} propose a time slot channel hopping scheme for \lorawan{} in industrial scenarios.
They show that their scheme increases network capacity to support up to 6,000 devices with the highest transmission period without collisions, if it uses all time slots.
They conclude that with the proposed technique \lorawan{} is compatible with the industrial requirements and similar to solutions as WirelessHART and ISA100.11a.
Muzammir \textit{et al.}~\cite{Muzammir:ISCAIE:2019} studies the feasibility of \lorawan{} for indoor applications, taking into account packet losses, data rates, and communication range. They conclude that \lorawan{} is suitable for indoor applications and encourage future works on scalability and reliability.
{Haxhibeqiri \textit{et al.}~\cite{Haxhibeqiri:2017:ETFA} tested \lora{} technology coverage in an industrial indoor area of 34,000 m$^2$.  With a simulation analysis using one gateway, all 6 SF and 3 different frequency channels, they could support up to 6,000 devices with a success probability of 0.9.}
Finally, Lentz \textit{et al.}~\cite{Lentz:IECON:2018} present a industrial real-time monitoring system operating with \lora{} technology, with sensors on a food-processing plant for preventive maintenance.

In a different line of research, Hoeller \textit{et al.}~\cite{Hoeller:ACS:2018} extend the model in~\cite{Georgiou:WCL:2017} by adding spatial and time diversity, concluding that simple message replication (RT) is often beneficial but may flood the network. Mo \textit{et al.}~\cite{Mo:WD:2016} investigate the same replication scheme but in \sigfox{},  showing that there is an optimal number of replications. Marcelis \textit{et al.}~\cite{Marcelis:IoTDI:2017} present DaRe, an application layer coding technique for \lorawan{}.
It features characteristics of the fountain and convolutional coding, where packets carry redundant information from previous packets. The technique improves \lorawan{} in exchange for latency. Sandell and Raza~\cite{Sandell:SJ:2018} analyze DaRe and show that latency increases exponentially with packet loss. Finally, Sánchez \textit{et al.}~\cite{Montejo-Sanchez:WCL:2019} present some coding schemes for LPWAN. They propose Coded Transmission - Independent (CT), which presents a better success rate than the methods in~\cite{Hoeller:ACS:2018} and~\cite{Marcelis:IoTDI:2017} for several scenarios. However, none of these works analyze the energy consumption impact of such replication schemes.

\subsection{Contribution}

We present a novel coded message replication scheme, Hybrid Transmission (HT), that generalizes RT~\cite{Hoeller:ACS:2018} and CT~\cite{Montejo-Sanchez:WCL:2019} and allows fine-tuning for different network densities and reliability targets. Moreover, in order to analyze the maximum number of devices, we numerically optimize the parameters of the proposed coded replication scheme. The new method provides important gains in terms of users in coverage compared to the methods in~\cite{Hoeller:ACS:2018} and~\cite{Montejo-Sanchez:WCL:2019}, while the advantage of the proposed scheme increases with the reliability target. On top of that, we build upon an energy consumption model~\cite{Casals:SENSORS:2017} to analyze the impact of replications on battery lifetime. Finally, we propose a slight change to \lorawan{} to avoid unnecessary energy waste with replication schemes.

The rest of this paper is organized as follows.
Section~\ref{sec:lora} describes \lora{} and \lorawan{}.
Section~\ref{sec:base_model} presents the adopted system model.
Section~\ref{sec:diversity} discusses some replication schemes for LPWAN.
Section~\ref{sec:hybrid} introduces and analyzes the novel HT coded replication scheme.
Section \ref{sec:results} presents numerical results.
Finally, Section \ref{sec:conclusion} concludes the paper.

\section{\lora{} Overview}
\label{sec:lora}

\lora{} is a proprietary sub-GHz PHY technology for long-range low-power communications that employs chirp spread spectrum. \lora{} uses six practically orthogonal SFs, which allows for a rate and range tradeoff~\cite{Semtech:2015}. \lorawan{} is the most popular network protocol using \lora{} as a PHY, being regulated by \lora{} Alliance~\cite{LoRaAlliance:2017}.
\lorawan{} presents a star topology where devices communicate in a single-hop with the gateway, which in turn has a standard IP connection with a network server. \lorawan{} uses ALOHA~\cite{Goldsmith:Book:WC:2005} as medium access control, exploiting \lora{} characteristics, enabling multiple devices to communicate at the same time using different SFs. 

\lorawan{} provides several configurations, limited by regional regulations, and based on SF, that varies from 7 to 12, and bandwidth ($B$), usually 125 kHz or 250 kHz for uplink.
Higher SFs extend the symbol duration, reducing the rate and, thus, increasing  robustness.
Features that reduce the rate, such as higher SF and lower bandwidth, also increase the Time-on-Air (ToA), which in turn increases the channel usage, ultimately impacting the collision probability~\cite{Hoeller:ACS:2018}.
After each transmission, the \lorawan{} device opens two receive windows, with duration $T_{rx1w}$ and $T_{rx2w}$, for downlink messages.
Moreover, since \lora{} is a type of frequency modulation, it presents the capture effect~\cite{Bor:MSWIM:2016}, enabling the receiver to decode the strongest signal when there is a collision, provided that this signal is at least $\theta$ dB above the interference. 

\begin{table}[t]
\centering
\caption{LoRa Uplink characteristics. Payload of 9 bytes, $B~=~125$~kHz, CRC and Header modes enabled.}
\label{tab:lora_sensi}
\begin{tabular}{@{}lcccc@{}}
\toprule
\textbf{SF} & Time-on-Air $t_j$ & $T_{rx1w}$ & $T_{rx2w}$ & SNR threshold $q_j$ \\ \midrule
7                            & 41.22 ms    & 12.29 ms  & 1.28 ms  & -6 dB               \\
8                            & 72.19 ms    & 24.58 ms  & 2.30 ms  & -9 dB               \\
9                            & 144.38 ms   & 49.15 ms  & 4.35 ms  & -12 dB              \\
10                           & 247.81 ms   & 98.30 ms  & 8.45 ms  & -15 dB              \\
11                           & 495.62 ms   & 131.07 ms & 16.64 ms & -17.5 dB            \\
12                           & 991.23 ms   & 262.14 ms & 33.02 ms & -20 dB              \\ \bottomrule
\end{tabular}
\end{table}
Finally, Table~\ref{tab:lora_sensi} presents the sensibility, signal-to-noise ratio (SNR) thresholds, and ToA and receive windows time for each SF.
It shows that the increase of the SF nearly halves SNR ($-3$dB). Note that all times are longer with higher SF.

\section{System Model}
\label{sec:base_model}

Similar to \cite{Georgiou:WCL:2017}, consider a circular region with radius $R$ meters and area $V = \pi R^2$, where, on average, $\bar{N}$ end-devices are uniformly deployed at random.
The positioning of such devices is described by six Poisson Point Processes (PPP) $\Phi_j$, one for each SF$j$, with density $\rho_j > 0$ in $V$ and 0 otherwise, where $\bar{N}=\rho V$.
The distance from $i$-th device to the gateway at the origin is $d_i$ meters.
Devices transmit using unslotted-ALOHA.
All nodes run the same application, while the ratio of time they transmit in an interval is the activity factor
    $p_j = \frac{t_j}{P}$,
where $P$ is the average message period, \textit{i.e.}, the time between transmissions, and $t_j$ is the time-on-air for SF$j$, according to Table~\ref{tab:lora_sensi}.
The activity factor vector is $p = [p_7, \dots, p_{12}]$, each element being the activity factor for a SF.
If nodes transmit a given amount of messages per day, due to different ToA per SF, nodes using SF12 would have a larger activity factor than nodes using SF7.
Note that the activity factor should respect the regional duty cycle restriction, which is 1\% for the most used configurations in Europe.
All devices use the same bandwidth $B$~Hz and the same total transmit power $\mathcal{P}_t$.

The model considers both path loss attenuation, $g(d_i)$, and Rayleigh fading, $h_i$.
We use an empirical path loss model of an indoor industrial environment in the sub-GHz band~\cite{Tanghe:WCL:2008}, $g(d_i) = \textup{PL}_0^{-1} \left( \tfrac{d_i}{d_0} \right)^{-\eta}$, where $\eta$ is the path loss exponent and $\textup{PL}_0$ is the loss at the reference distance $d_0$ meters.
Given a signal $s_1$ transmitted by a \lora{} end-device, the received signal at the gateway, $r_1$, is the sum of the attenuated transmitted signal, interference and noise,
\begin{equation}
    r_1 = \sqrt{\mathcal{P}_t g(d_1)} h_1 s_1 + \sum_{k\in\Phi_j} {\sqrt{\mathcal{P}_t g(d_k)} h_k s_k} + w \label{eqn:r1},
\end{equation}
where $k$, in the summation, iterates over the active nodes indicated by the PPP, and $w$ is the additive white Gaussian noise with zero mean and variance $\mathcal{N} = -174 + \textup{NF} + 10 \log(B)$~dBm, while $\textup{NF}$ is the receiver noise figure.
An outage takes place either if there is no connection between a node and the gateway or if there is a collision.

\subsection{Outage Condition 1: Disconnection}

The disconnection probability is the complement of the connection probability denoted by $H_1$.
The connection probability depends directly on the distance between a node and the gateway.
A node is considered connected if the SNR of the received signal is above the SNR threshold in Table~\ref{tab:lora_sensi}.
The connection probability can be written as
\begin{equation}
    H_{1} = \mathbb{P}[\textup{SNR} \geq q_j ~|~ d_1], \label{eqn:h1_1}
\end{equation}
where $q_j$ is the SNR threshold for SF$j$, and $d_1$ is the distance between that node and the gateway.
The node at $d_1$ uses SF$j$.
Assuming Rayleigh fading implies exponentially distributed instantaneous SNR because $|h_1|^2 \sim \textup{exp}(1)$. Thus, $H_1$ is
\begin{equation}
    H_{1} = \mathbb{P} \left [ |h_1|^2 \geq \frac{\mathcal{N} q_j}{\mathcal{P}_t g(d_1)} ~\Big|~ d_1 \right ] = \textup{exp} \left ( - \frac{\mathcal{N} q_j}{\mathcal{P}_t g(d_1)} \right ) \label{eqn:h1_2}.
\end{equation}

\subsection{Outage Condition 2: Collision}

A collision happens when packets are simultaneously transmitted using the same SF.
Due to \lora{} capture effect, the receiver recovers a packet if its power is at least $\theta$~times above the interference.
We model the capture probability, \textit{i.e.}, the complement of the collision probability, as in~\cite{Hoeller:ACS:2018}.
First, we define the Signal-to-Interference Ratio (SIR) of a packet from a node $d_1$ meters from the gateway using SF$j$ as
\begin{align}
    \textup{SIR} & = \frac{\mathcal{P}_t|h_1|^2 g(d_1)}{\sum_{k\in\Phi_j} \mathcal{P}_t |h_k|^2 g(d_{k})} = \frac{|h_1|^2 d_1^{-\eta}}{\sum_{k\in\Phi_j} |h_k|^2 d_{k}^{-\eta}}.
\end{align}
Then, due to the capture effect, the probability of a successful reception in the presence of interference is
\begin{align}
    Q_{1} &= \mathbb{P}\left[\textup{SIR} > \theta ~|~ d_1 \right] \nonumber \\
    &= \mathbb{P}\left[|h_1|^2 > \theta d_1^\eta \sum_{k\in\Phi_j} |h_k|^2 d_{k}^{-\eta} ~\Big|~ d_1 \right] \nonumber \\
    &= \mathbb{E}_{|h_k|^2, \Phi} \left[ \prod_{k\in\Phi_j} \textup{exp}\left ( - \theta d_1^\eta |h_k|^2 d_{k}^{-\eta} \right ) \right].
\end{align}
Since $|h_k|^2$ is also exponentially distributed as  $|h_1|^2$, then
\begin{align}
    Q_{1} &= \mathbb{E}_{\Phi_j} \left[ \prod_{k\in\Phi_j} \frac{1}{1 + \theta d_1^\eta d_{k}^{-\eta}} \right].
\end{align}
Following~\cite{Hoeller:ACS:2018}, we use the probability generating functional of the product over PPPs, where $\mathbb{E}[\prod_{x\in\Phi} f(x)] = \textup{exp}(-\alpha_j \int_{R^2} 1-f(x) ~\textup{d}x)$ with $\alpha_j = 2 \rho_j p_j$ as the PPP density and $d_k$ converted to polar coordinates. Note that we double $\alpha_j$ due to the use of unslotted-ALOHA~\cite{Berioli:FTN:2016}. Thus
\begin{align}
    Q_{1} &= \textup{exp}\left ( -4\pi \rho_j p_j \int_{0}^{R} \frac{\theta d_1^\eta d_k^{-\eta}}{1 + \theta d_1^\eta d_k^{-\eta}} d_k~\textup{d}d_k \right ), \label{eqn:q1a}
\end{align}where $R$ is the network radius and $p_j$ is the activity factor related to the SF$j$ used by the node at $d_1$.

Finally, by using the definition of the Gauss Hypergeometric function $_2F_1(a,b;c;z)$~\cite{Daalhuis:Chapter:2010} we have that
\begin{align}
    Q_{1} &= \textup{exp} \left [ -2 ~ \pi R^2 ~ \rho_j p_j ~ _2F_1 \left(1,\frac{2}{\eta};1+\frac{2}{\eta}; -\frac{R^{\eta}}{\theta d_1^\eta} \right) \right ] . \label{eqn:q1}
\end{align}

\subsection{Coverage Probability}

The coverage probability is the probability that a node can communicate with the gateway.
Since the connection probability in \eqref{eqn:h1_2} and the probability of a successful reception in the presence of interference in~\eqref{eqn:q1} are treated as independent probabilities, we lower bound the coverage probability as $H_1Q_1$~\cite{Beltramelli:WIMOB:2018}.
Therefore, an upper bound to the link outage probability is defined as $\mathcal{O}_1 = 1-H_1Q_1$.

\section{Previous Time Diversity Schemes} \label{sec:diversity}

This section presents two message replication schemes for LPWANs. For the sake of clearness, we will call information message the first uncoded message transmitted inside a period. All the other messages we call redundant messages.

\subsection{Replication Transmission (RT)}\label{sec:RT}

Hoeller \textit{et al.}~\cite{Hoeller:ACS:2018} considered RT, simple message replication, where a message is sent $m$ times during one period, disregarding any downlink channel acknowledgment.
Message replicas are separated in time to respect the time coherence of the channel, ensuring different channel realizations.
Notice that the increase of transmissions inside one period affects $Q_1$ since it increases the collision probability.
$Q_{1}$ takes the number of interfering nodes as a Poisson distributed variable with mean $v_j=p_j\rho_j V$.
With $m$ messages sent in each period, \textit{i.e.}, $M=m$ messages in total, the channel usage increases $M$ times.
Thus, it is necessary to adjust the mean of the Poisson distribution proportionally, resulting in $v_M=Mp_j\rho_j V$.
We denote by $Q_{1,M}$ the capture probability $Q_1$ with activity factor $Mp_j$.
Note that $Q_1=Q_{1,1}$.
The link outage probability is
\begin{equation}
    \mathcal{O}_M = 1 - H_1Q_{1,M}.
\end{equation}
The probability of failing to decode all $M$ packets is
\begin{equation} \label{rt_outage}
    \mathcal{O}_{\textup{RT}}(M) = (\mathcal{O}_{M})^{M}.
\end{equation}

\subsection{Coded Transmission (CT)}

The RT scheme is the simplest way to achieve time diversity.
However, lower outages can be obtained by replicating coded messages, which are combinations of previous information messages.
The receiver decodes the coded messages to recover lost information, provided that a certain amount of messages is successfully received.
Marcelis \textit{et al.}~\cite{Marcelis:IoTDI:2017} and Montejo-Sánchez \textit{et al.}~\cite{Montejo-Sanchez:WCL:2019} present cases of embedded redundancy, where packets carry both systematic and parity parts.
Montejo-Sánchez \textit{et al.}~\cite{Montejo-Sanchez:WCL:2019} shows, however, that embedded schemes tend to be less reliable, less energy efficient, and provide less coverage than independently coded packet transmissions, as is the case of their proposed independent Coded Transmission (CT) scheme.

The central concept of CT is to combine different messages using linear operations (\textit{e.g.}, XOR) and send them as independent packets in addition to the information messages.
In CT, assuming two information messages A and B that are independently transmitted, an additional coded message $\textup{A}\oplus\textup{B}$ is also transmitted as redundancy.
If any two of these three messages are successfully received, the decoder can recover both A and B information messages.

The CT scheme has one parameter, $n$, which is the number of coded transmissions per information message.
If $n=0$ then no replication takes place.
For example, when $n=1$, the $k$-th information message is followed by a coded transmission combining the $k$-th and $(k-1)$-th messages.
If $n=2$, the $k$-th information message is followed by a coded transmission combining the $k$-th and $(k-1)$-th messages, and a coded transmission combining the $k$-th and $(k-2)$-th messages.
Thus, each information packet is followed, within the period, by $n$ redundant packets.
In an infinite sequence of transmissions, each message appears in a window of $2n$ transmissions in coded messages plus the original information one.
Theoretically, there is an infinite amount of combinations of received messages allowing for the recovery of a lost information message.
To limit the latency, however, the decoding window of each message is restricted to $\pm 3$, \textit{i.e.}, from ($k-3$)-th to ($k+3$)-th.
This restriction also allows the deduction of the following closed-form outage probability expression of CT which depends on $M=n+1$~\cite{Montejo-Sanchez:WCL:2019}:
\begin{equation}
    \mathcal{O}_{\textup{CT}}(n) = \mathcal{O}_M^{2n+1}(1+ \mathcal{O}_M+ \mathcal{O}_M^2 - 5\mathcal{O}_M^3 + 4\mathcal{O}_M^4 - \mathcal{O}_M^5)^{2n}.\label{cti_outage}
\end{equation}

\section{Hybrid Replication Scheme}
\label{sec:hybrid}

The CT scheme presents a significant improvement compared to RT and other embedded redundancy methods~\cite{Montejo-Sanchez:WCL:2019}.
Note that, for RT, increasing the number of replicas has a beneficial impact in the exponent of the outage probability, but also a negative impact in the number of messages, what increases the number of collisions.
In CT, the successive replicas are combined before transmission. The use of linear combinations (e.g., XOR operation) of previous messages increases the redundancy with limited cost in terms of spectral efficiency~\cite{Montejo-Sanchez:WCL:2019}.
However, these linear combinations generate dependence on multiple information for the successful decoding of an encoded message, which is harmful for high link outage probabilities.
Consequently, a Hybrid Transmission (HT) scheme capable of using replicas of uncoded messages and redundancy based on linear combinations with previous messages should outperform RT and CT.
The proposed HT scheme allows for the transmission of uncoded redundant messages only, of coded messages only, or a combination of uncoded and coded redundant replications.

The proposed HT scheme has three parameters: $m$ is the number of uncoded message replicas, as in RT; $n$, as in CT, is the number of different coded messages; and $r$ is the number of replicas of each coded message.
For instance, when $n=2$, in the CT scheme a node transmits the $k$-th information message followed by two coded messages: $k$-th $\oplus$ ($k-1$)-th and $k$-th $\oplus$ ($k-2$)-th. In HT, with $n=2$, $m=2$ and $r=3$, for example, a node sends the $k$-th uncoded message twice ($m=2$), and then sends  two differently coded messages ($n=2$) three times each ($r=3$). Then, in HT $M=m+nr$ is the number of messages containing either information or redundancy in a period.

\begin{table}[tb]
\centering
\caption{Expansion of possible successful decoding events of HT from $k-3$ to $k+3$ with $n=1$.}
\label{tab:events}
\scalebox{0.86}{
\begin{tabular}{|c|c|c|c|c|c|c|c|c|c|c|c|c|c|}
\hline
\multicolumn{2}{|c|}{k-3} & \multicolumn{2}{c|}{k-2} & \multicolumn{2}{c|}{k-1} & \multicolumn{2}{c|}{k} & \multicolumn{2}{c|}{k+1} & \multicolumn{2}{c|}{k+2} & \multicolumn{2}{c|}{k+3} \\ \hline
\textbf{$\mathcal{M}$} & \textbf{$\mathcal{R}$} &\textbf{$\mathcal{M}$} & \textbf{$\mathcal{R}$} &\textbf{$\mathcal{M}$} & \textbf{$\mathcal{R}$} &\textbf{$\mathcal{M}$} & \textbf{$\mathcal{R}$} &\textbf{$\mathcal{M}$} & \textbf{$\mathcal{R}$} &\textbf{$\mathcal{M}$} & \textbf{$\mathcal{R}$} &\textbf{$\mathcal{M}$} & \textbf{$\mathcal{R}$} \\ \hline
 &  &  &  &  &  & S &  &  &  &  &  &  &  \\ \hline
 &  &  &  & \cellcolor[HTML]{87CEFA}S &  & F & \cellcolor[HTML]{87CEFA}S &  &  &  &  &  &  \\ \hline
 &  &  &  &  &  & F &  & \cellcolor[HTML]{34FF34}S & \cellcolor[HTML]{34FF34}S &  &  &  &  \\ \hline
 &  & \cellcolor[HTML]{87CEFA}S &  & \cellcolor[HTML]{87CEFA}F & \cellcolor[HTML]{87CEFA}S & F & \cellcolor[HTML]{87CEFA}S &  &  &  &  &  &  \\ \hline
 &  &  &  &  &  & F &  & \cellcolor[HTML]{34FF34}F & \cellcolor[HTML]{34FF34}S & \cellcolor[HTML]{34FF34}S & \cellcolor[HTML]{34FF34}S &  &  \\ \hline
\cellcolor[HTML]{87CEFA}S &  & \cellcolor[HTML]{87CEFA}F & \cellcolor[HTML]{87CEFA}S & \cellcolor[HTML]{87CEFA}F & \cellcolor[HTML]{87CEFA}S & F & \cellcolor[HTML]{87CEFA}S &  &  &  &  &  &  \\ \hline
 &  &  &  &  &  & F &  & \cellcolor[HTML]{34FF34}F & \cellcolor[HTML]{34FF34}S & \cellcolor[HTML]{34FF34}F & \cellcolor[HTML]{34FF34}S & \cellcolor[HTML]{34FF34}S & \cellcolor[HTML]{34FF34}S \\ \hline
\end{tabular}
}
\end{table}

\newtheorem{theorem}{{\bf Theorem}}
\newtheorem{lemma}{{\bf Lemma}}
\newtheorem{proposition}{Proposition}
\newtheorem{corollary}{{\bf Corollary}}[theorem]
\newtheorem{definition}{Definition}
\bigskip
\begin{lemma}\label{lem:1}
\textit{The outage probability for HT with $n=1$ is a function of two independent sets of events with probability}
\begin{align} \label{eq:hyb_Etab}
    \mathcal{E}_{m,n,r} &={} (1-\mathcal{O}_M^m)(1-\mathcal{O}_M^r) + \mathcal{O}_M^m(1-\mathcal{O}_M^m)(1-\mathcal{O}_M^r)^2 \nonumber\\
    & + \mathcal{O}_M^{2m}(1-\mathcal{O}_M^m)(1-\mathcal{O}_M^r)^3.
\end{align}
\end{lemma}
\begin{proof}
Consider Table \ref{tab:events}, where each cell indicates whether the decoding of a message was successful (S) or failed (F).
Each wide column represents a period, going from $(k-3)$-th to $(k+3)$-th.
The narrow columns represent packets containing either original messages ($\mathcal{M}$) or coded messages ($\mathcal{R}$).
Each $\mathcal{M}$ or $\mathcal{R}$ column represents events already accounting for the repetitions ($m$ or $r$) of either type of message.
Thus, an F on any cell represents an outage event.
Therefore, the outage probability of an original message is $\mathbb{P}(\mathcal{M}=\textup{F})=\mathcal{O}_M^m$, and the outage probability of a coded message is $\mathbb{P}(\mathcal{R}=\textup{F})=\mathcal{O}_M^r$.
There are two independent sets of events with the same occurrence probability that allow for the decoding of the $k$-th message if the original transmission fails.
The first set, which consists of the events that are a linear combination of the $(k-1)$-th message, is denoted $E_{H,1}$ and highlighted in blue.
The second set is formed by the events that are a linear combination of the $(k+1)$-th message, denoted $E_{H,2}$ and highlighted in green.
Using Table~\ref{tab:events}, we can determine the probability of these events as $\mathbb{P}[E_{H,1}]=\mathbb{P}[E_{H,2}]=\mathcal{E}_{m,n,r}$.
Therefore, the outage probability of HT for $n=1$ can be determined as the union of the events in Table \ref{tab:events}: the original packet transmission (first line in the table) and the two independent sets of events, considering that the original transmission fails. Thus,
\begin{align} \label{eq:hyb_1n}
    \mathcal{O}_{\textup{HT}}(m,1,r) &={} 1 - \left [ (1-\mathcal{O}^m_M) + \mathcal{O}^m_M\mathbb{P} \left ( E_{H,1} \bigcup E_{H,2} \right ) \right ] \nonumber \\
 &={} \mathcal{O}^m_M - \mathcal{O}^m_M(2\mathcal{E}_{1,m,r} - \mathcal{E}_{1,m,r}^2) \nonumber \\
 &={}\mathcal{O}^m_M(\mathcal{O}^{3m}_M+\mathcal{O}^r_M+\mathcal{O}^{m+r}_M-2\mathcal{O}^{2(m+r)}_M \nonumber \\
& -\mathcal{O}^{3(m+r)}_M + \mathcal{O}^{2m+r}_M - 3\mathcal{O}^{3m+r}_M - \mathcal{O}^{m+2r}_M \nonumber \\
& + 3\mathcal{O}^{3m+2r}_M + \mathcal{O}^{2m+3r}_M)^2.
\end{align}
\end{proof}

\bigskip
\begin{lemma}\label{lem:2}
\textit{There are $2n$ independent sets of events, where}
\begin{equation}
    \mathbb{P}(E_{H,j}) = \mathcal{E}_{m,n,r}, \forall j \in \{ 1, \ldots, 2n \}.
\end{equation}
\end{lemma}
\begin{proof}
Each independent set of events is related to a linear combination of a message. When $n=1$, we have the events that are a linear combination of the $(k-1)$-th and $(k+1)$-th messages. With $n=2$, we also have the events related to the $(k-2)$-th and $(k+2)$-th messages. This set expansion follows as we have sets from the $(k-n)$-th to the $(k+n)$-th messages, totaling $2n$ independent sets of events.
\end{proof}

\bigskip
\begin{theorem}
\textit{The outage probability of HT is}
\begin{align}\small \label{eqn:HT_n_outage}
    \mathcal{O}_{\textup{HT}}(m,n,r) = \mathcal{O}_M^{m(2n+1)} (F_{m,n,r})^{2n},
\end{align}
where
\begin{align}
    F_{m,n,r} &= \mathcal{O}_M^{2m} + (1-\mathcal{O}_M^m)(\mathcal{O}_M^{m+3r} - \mathcal{O}_M^{2r} - 3\mathcal{O}_M^{m+2r}) \nonumber \\
    & +\mathcal{O}_M^r(1+\mathcal{O}_M^{-m} +\mathcal{O}_M^m - 3\mathcal{O}_M^{2m}).
\end{align}
\end{theorem}
\begin{proof}
From Lemma \ref{lem:2}, there are $2n$ independent sets of events, while from the proof of Lemma~\ref{lem:1}, the outage probability of HT depends on the union of the $2n$ events. Thus,
\begin{align}\small \label{eq:HT_out1}
  \mathcal{O}_{\textup{HT}}(m,n,r) &= 1 - \left [(1-\mathcal{O}^m_M) + \mathcal{O}^m_M \mathbb{P} \left ( \bigcup_{j=1}^{2 n} E_{H,j} \right ) \right ].
 \end{align}
Considering the probability of the union of sets of events of the last term, there is an outage only when all $2n$ $E_{H,j}$ events fail.
Thus, we write this success probability as its complement
\begin{equation} \label{eq:hyb_pe}
    \mathbb{P} \left ( \bigcup_{j=1}^{2 n} E_{H,j} \right ) = 1-(1-\mathcal{E}_{m,n,r})^{2n}.
\end{equation}
Applying~(\ref{eq:hyb_Etab}) and~(\ref{eq:hyb_pe}) to~(\ref{eq:HT_out1}) and simplifying yield~(\ref{eqn:HT_n_outage}).
\end{proof}

\bigskip
Note that if $n=0$ then there is no coding and only simple replicas are sent, making $r$ meaningless. Thus, when $n=0$ (\ref{eqn:HT_n_outage}) reduces to (\ref{rt_outage}), leading to Corollary \ref{cor:ht-rt}.

\bigskip
\begin{corollary} \label{cor:ht-rt}
\textit{The HT scheme is a generalization of RT since} $\mathcal{O}_{\textup{HT}}(m,0,r) = \mathcal{O}_{RT}(m)$.
\end{corollary}

\bigskip
When $m=r=1$  HT transmits the information message and $n$ coded messages, the same idea as CT, so that (\ref{eqn:HT_n_outage}) reduces to~(\ref{cti_outage}), leading to  Corollary \ref{cor:ht-cti}.

\bigskip
\begin{corollary} \label{cor:ht-cti}
\textit{The HT scheme is a generalization of CT since} $\mathcal{O}_{\textup{HT}}(1,n,1) =  \mathcal{O}_{CT}(n).$
\end{corollary}

\bigskip
\begin{theorem}\label{the:2}
\textit{The outage probability of HT is never larger than that of RT and CT, \textit{i.e.},}
\begin{equation}
    \mathcal{O}_{\textup{HT}}(m,n,r) \leq  \mathcal{O}_{\textup{CT}}(n),  \mathcal{O}_{\textup{RT}}(m).
\end{equation}
\end{theorem}
\begin{proof}
The proof comes from Corollaries~\ref{cor:ht-rt} and~\ref{cor:ht-cti}.
\end{proof}
\bigskip

Figure~\ref{fig:cmp_HT_M4} illustrates Theorem~\ref{the:2}, showing the final outage probability (after decoding the replications) versus the link outage probability (before decoding the replications), for three possible configurations of the HT scheme with $M=4$, including those equivalent to CT and RT.
In this example, the CT-specific configuration performs better for low link outage probability (below $\approx 0.4$), and the RT-specific configuration performs slightly better for link outage probabilities very close to 1. For all other values of $\mathcal{O}_M$, the HT-exclusive configuration outperforms the others. Notice that HT is, thus, always better than or equivalent to RT and CT, since it is possible to set HT to mimic the other schemes.

\begin{figure}[t]
    \centering
    \includegraphics{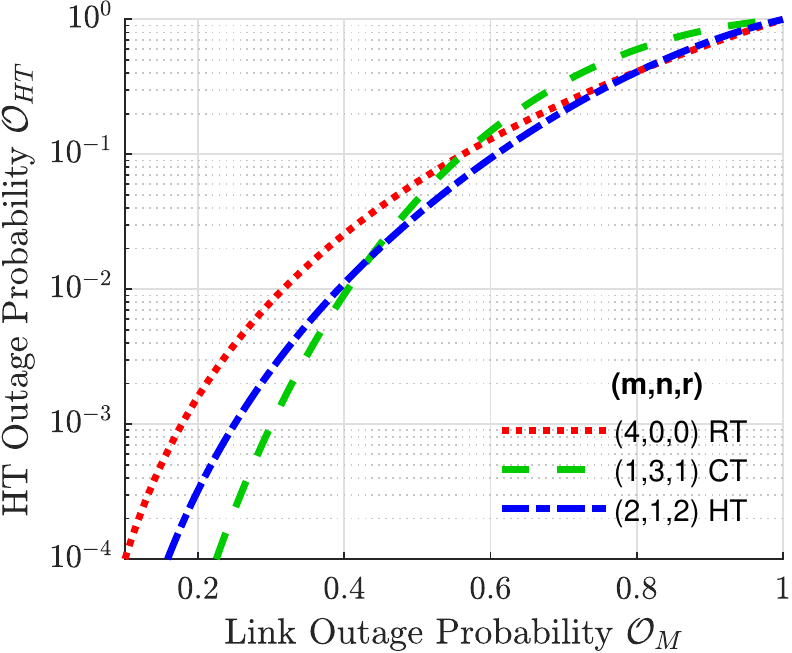}
    \caption{Final outage probability (after replication) versus link outage probability (before replication) for different HT configurations with $M=4$.}
    \label{fig:cmp_HT_M4}
\end{figure}

Regarding implementation, devices do not implement any optimization routine apart from acquiring and transmitting data. We should only include simple XOR operations and some buffers at each end-device with size depending on parameter $n$.
The network server is in charge of all processing with total control over the network, \textit{i.e.}, it has information on all active devices, including their SF and configuration of the replication scheme ($m$, $n$, $r$).
The network server should configure the devices in the join procedure or reconfigure through downlink MAC commands if it deems necessary.
It is an implementation choice if the configuration information is in the header or payload as a new message type.
However, both cases provide a minimum increase in the message length.
For example, with 8 extra bits in a header it is possible to store information about $m$ and $r$ varying from 1 to 8, and $n$ from 1 to 4, covering a huge set of different HT configurations.

\subsection{Average Number of Devices}

As stated in \cite{Candell:IEM:2018}, the number of devices is a key requirement of Industrial IoT applications.
However, it is also important to cover these devices with adequate reliability levels.
We propose to find the maximum number of users that guarantees the minimum reliability level $\mathcal{T}$.
To do so, we must ensure that all devices in the network are above $\mathcal{T}$. Thus, we evaluate this metric considering the worst case position in the network, \textit{i.e.}, at the border ($d_1 = R$).
From~\eqref{eqn:q1}, considering $d_1=R$ and the activity factor increased by $M$, we have that
\begin{align}
    Q_{M} &= \textup{exp} \left [ -2 ~ \pi R^2 ~ \rho_j Mp_j ~ _2F_1 \left(1,\frac{2}{\eta};1+\frac{2}{\eta}; -\frac{1}{\theta} \right) \right ] .
\end{align}
Isolating the node density, 
\begin{equation}
    \rho_j = - \frac{\textup{ln}\left(\frac{1-\mathcal{O}_M}{H_1}\right)}{2~\pi R^2 ~ M p_j ~ _2F_1 \left(1,\frac{2}{\eta};1+\frac{2}{\eta}; -\frac{1}{\theta} \right) },
\end{equation}
where we have $Q_M = \frac{1-\mathcal{O}_{M}}{H_{1}}$, recalling that ${O}_{M}$ is the link outage probability and ${H_1}$ is the connection probability. Note that since we are investigating the outage at a fixed distance, the connection probability $H_1$ is a constant, irrespective of $\rho$.
Finally, since the disk total area is $V=\pi R^2$ and the average number of nodes is $\bar{N}=\rho V$, we have that
\begin{equation}
    \bar{N_j} = - \frac{\textup{ln}\left(\frac{1-\mathcal{O}_M}{H_1}\right)}{2~M p_j  ~ _2F_1 \left(1,\frac{2}{\eta};1+\frac{2}{\eta}; -\frac{1}{\theta} \right) }. \label{eq:N_Q}
\end{equation}
Note that $\bar{N_j}$ denotes the number of devices using SF$j$, while $\bar{N}$ is the total number of devices in the network.
Given~\eqref{eq:N_Q} and the outage expressions of the replication schemes, we aim to find the maximum number of nodes supported by the network that guarantees a minimum reliability level.
To do so, we need do invert all the equations and find the configuration that minimizes the outage, and thus maximize $\bar{N}$. Since CT and HT equations are quite complicated, we evaluate them numerically enumerating all possible configurations up to a maximum number of message copies ($M$).

\subsection{Energy Consumption Model}

The lifetime of a device battery is an important concern for Industrial IoT applications.
SF plays an important role in \lorawan{} energy consumption.
Besides the usual impact of transmission power, transmission period, and payload length, SF increases signal length and therefore greatly increases energy expenditure because higher SFs present lower data rates, extending transmission time and, thus consuming more energy.
When implementing the presented replication schemes, we expect a trade-off between reliability and battery lifetime.

To quantify the impact of the replication schemes, we use an energy consumption model for unacknowledged \lorawan{} presented in \cite{Casals:SENSORS:2017}, which considers the energy consumption and time duration of 11 operating states of a \lorawan{} device.
The average current consumption of a device as a function of time spent in each state is
\begin{equation}
    I_{avg}= \frac{M}{P}\left(\sum^{N_{states}-1}_{i=1}T_iI_i+T_{sleep}I_{sleep}\right),
\end{equation}
where $M$ is the number of message copies inside a period, $N_{states}=11$, $T_i$ and $I_i$ are respectively the duration and current consumption of state $i$ in Table \ref{tab:energy} and
\begin{equation}
    T_{sleep} = P-M\sum^{N_{states-1}}_{i=1} T_i. \label{eqn:sleep}
\end{equation}
Moreover, we calculate the theoretical lifetime of a battery-operated device as
\begin{equation}
    T_{lifetime} = \frac{C_{battery}}{I_{avg}},
\end{equation}
where $C_{battery}$ is the battery capacity.

\begin{table}[t]
\centering
\caption{Energy consumption states of \lorawan{}~\cite{Casals:SENSORS:2017}.}
\label{tab:energy}
\begin{tabular}{@{}llllll@{}}
\toprule
\multirow{2}{*}{\textbf{State}} & \multirow{2}{*}{\textbf{Description}} & \multicolumn{2}{l}{\textbf{Duration}} & \multicolumn{2}{l}{\textbf{Current}} \\ \cmidrule(l){3-6} 
 &  & \textbf{Variable} & \textbf{Value (ms)} & \textbf{Variable} & \textbf{Value (mA)} \\ \midrule
1 & Wake up & $T_{wu}$ & 168.2 & $I_{wu}$ & 22.1 \\
2 & Radio preparation & $T_{pre}$ & 83.8 & $I_{pre}$ & 13.3 \\
3 & Transmission & $t$ & Table \ref{tab:lora_sensi} & $I_{tx}$ & 83.0 \\
4 & Radio off & $T_{off}$ & 147.4 & $I_{off}$ & 13.2 \\
5 & Postprocessing & $T_{post}$ & 268.0 & $I_{post}$ & 21.0 \\
6 & Turn off sequence & $T_{seq}$ & 38.6 & $I_{seq}$ & 13.3 \\
7 & Wait 1st window & $T_{w1w}$ & 983.3 & $I_{1w}$ & 27.0 \\
8 & 1st receive window & $T_{rx1w}$ & Table \ref{tab:lora_sensi} & $I_{rx1w}$ & 38.1 \\
9 & Wait 2nd window & $T_{w2w}$ & $1-T_{rx1w}$ & $I_{2w}$ & 27.1 \\
10 & 2nd receive window & $T_{rx2w}$ & Table \ref{tab:lora_sensi} & $I_{rx2w}$ & 35.0 \\
11 & Sleep & $T_{sleep}$ & \eqref{eqn:sleep} & $I_{sleep}$ & $45 \times 10^{-3}$ \\ \bottomrule
\end{tabular}
\end{table}

From this model, we see that a considerable amount of energy is spent in the downlink receive windows the device opens after each transmission -- in the case of replications, the number of receive windows increases with $M$.
Note that unsynchronized Class A nodes must open the receive window for its entire duration, even if it does not detect any downlink message.
On top of that, as indicated in~\cite{CENTENARO:PIMRC:2017}, the excessive use of \lorawan{} downlink can severely worsen network performance.
Based on the above facts, we propose that devices only open receive windows when transmitting the last message, avoiding the excessive receive windows when sending redundant replications.
With this, the new average current consumption $I_{avg2}$ is
\begin{equation}
    I_{avg2}\!=\!\! \frac{M}{P}\!\!\left(\!\sum^{N_{states}-5}_{i=1}\!\!\! T_iI_i \!+\! \!\!\sum^{N_{states-1}}_{i=7}\!\!\!T_iI_i+T_{sleep2}I_{sleep}\!\right)\!\!,
\end{equation}
where $T_{sleep2}$ is the new sleep duration
\begin{equation}
    T_{sleep2} = P - M\sum^{N_{states}-5}_{i=1} T_i - \sum^{N_{states}-1}_{i=7} T_i.
\end{equation}
Note that the Sleep state is a low power mode, consuming around 1000 times less energy than the other states.

\section{Numerical Results}\label{sec:results}

We evaluate the proposed scheme in terms of success probability, energy consumption, and the maximum number of users while meeting a certain reliability target.
We compare HT with the other replication schemes RT and CT.

We parameterize the path loss  with the empirical data from~\cite{Tanghe:WCL:2008}, \textit{i.e.}, path loss exponent $\eta=3.51$, path loss reference $PL_0=55.05$~dB, and reference distance $d_0=15$ meters.
We also consider network radius $R=200$ meters as proposed by~\cite{Luvisotto:WCMC:2018}, SIR level $\theta=1$~dB as measured in~\cite{Croce:CL:2018,Mahmood:TII:2019}, and total transmit power $\mathcal{P}_t=11$~dBm as considered by~\cite{Casals:SENSORS:2017}.
The \lora{} transceiver uses bandwidth $B=125$~kHz and receiver noise figure $\textup{NF}=6$~dB.
We assume that all devices run the same application with one information message every 10 minutes ($P=600$ seconds), which means 144 messages per day on average, resulting on activity factors $\{69, 120, 241, 413, 826, 1652\} \times10^{-6}$ according to the respective SF\{7-12\}.
We consider reliability targets of $\mathcal{T} \in \{0.99, 0.999\}$ and restrict the search space of the optimization procedure to $1\leq M \leq 10$ message copies per period.
Note that we limited SF12 up to 6 message copies to adhere to regional duty cycle constraints of 1\%.
We consider a battery capacity $C_{battery} = 2400$~mAh, what is typical of AA batteries.
Finally, we call DT (Direct Transmission) the baseline scenario with no replication scheme, as in \cite{Montejo-Sanchez:WCL:2019}, and HT$^*$ the HT scheme with number of messages copies $M$ restrained to at maximum the same number as CT, therefore ensuring a HT$^*$ configuration using the same amount of resources as CT.

Figure~\ref{fig:users_out} compares the outage probability using SF7 for different $\bar{N_7}$.
We consider a single SF because including them into one curve would cause confusion.
Nevertheless, similar plots for other SFs present the same tendency.
Also, the numbers next to the curves represent the message copies ($M$) used to achieve the results in that area.
First, we see that any replication scheme performs significantly better than DT.
We also see that optimal RT requires more message copies and still has a higher outage than the other schemes.
Optimal CT is the scheme that requires fewer message copies and outperforms RT.
Optimal HT outperforms all other schemes at the cost of more messages compared to CT.
However, optimal HT$^*$ can still outperform CT using the same amount of resources.
\begin{figure}[t]
    \centering
    \includegraphics{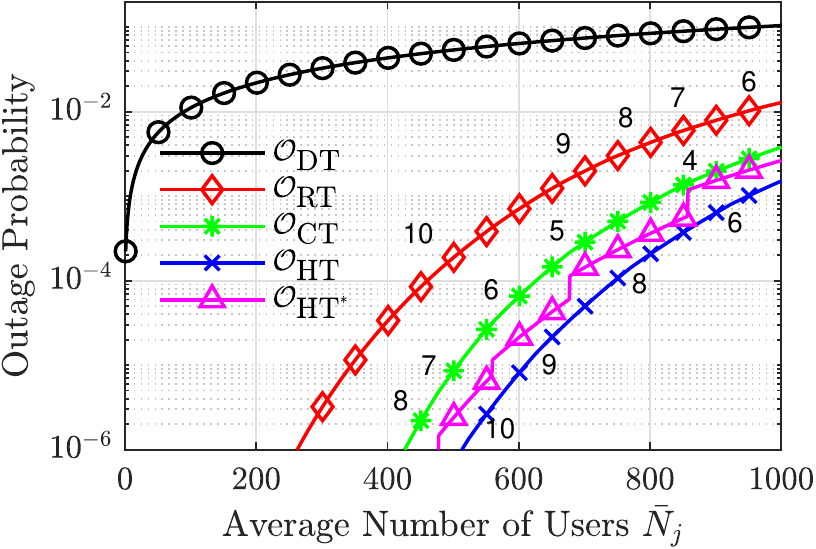}
    \caption{Outage probability of SF7 for each scheme using optimal configuration for different number of users $\bar{N_7}$.}
    \label{fig:users_out}
\end{figure}
\begin{figure}[t]
    \centering
    \includegraphics{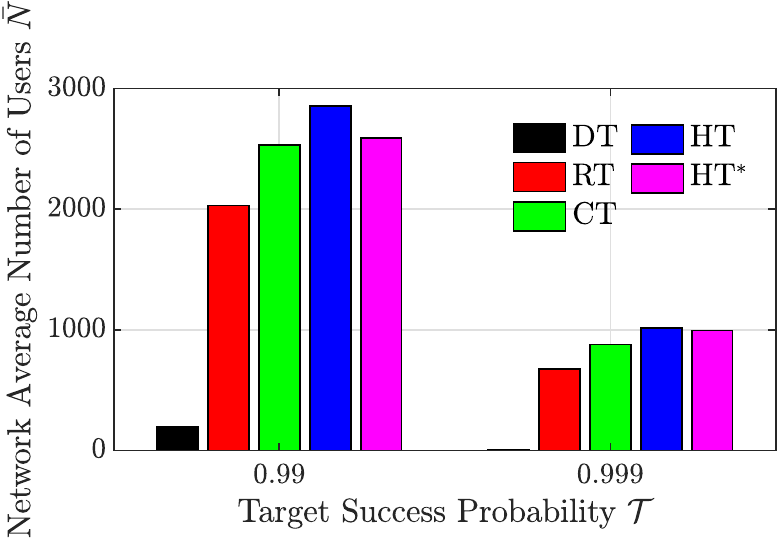}
    \caption{Average number of nodes $\bar{N}$ for each scheme and reliability targets $\mathcal{T} \in \{0.99, 0.999\}$.}
    \label{fig:bars_users}
\end{figure}

Figure~\ref{fig:bars_users} shows the average number of users in a network to maintain a target success probability $\mathcal{T}$ for different optimized schemes.
The parameters RT and CT are detailed in Table~\ref{tab:rt_cti_optimal}, while HT parameters are in Table~\ref{tab:ht_optimal}.
Again, we see that HT outperforms the other schemes at the expense of some more message copies.
However, HT$^*$ still has better results than CT using the same resources.
Again we see that DT has the worst performance.
With higher reliability, $\mathcal{T}=0.999$, we see that the difference from CT and HT$^*$ to HT is smaller than for $\mathcal{T}=0.99$.
This happens because with higher reliability levels, CT tends to perform better, as Figure~\ref{fig:cmp_HT_M4} showed.

Figure~\ref{fig:proposed-lifetime} presents battery lifetime as a function of message copies for SFs 7 and 12, considering regular \lorawan{} and the proposed adapted \lorawan{} with less receive windows.
Increasing the number of messages greatly impacts the battery lifetime, but the impact reduces with the proposed protocol modification.
Also, devices using SF7 tend to have higher battery lifetime than SF12, due to SF12 longer ToA.

\begin{figure}[t]
    \centering
    \includegraphics{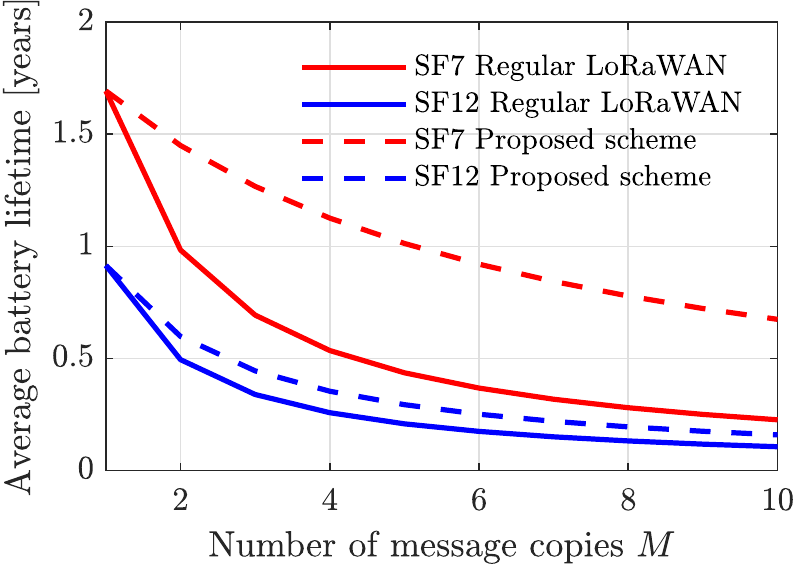}
    \caption{Average device lifetime using SF7 and SF12 for default \lorawan{} and the proposed modified protocol.}
    \label{fig:proposed-lifetime}
\end{figure}

\begin{table}[t]
\centering
\caption{RT and CT optimal number of message copies $M$ for each reliability target and spreading factor.}
\label{tab:rt_cti_optimal}
\begin{tabular}{@{}lccccccc@{}}
\toprule
\textbf{Scheme}     & \textbf{$\mathcal{T}$}     & \textbf{SF7}            & \textbf{SF8}            & \textbf{SF9}            & \textbf{SF10}          & \textbf{SF11}          & \textbf{SF12} \\ \midrule
\multirow{2}{*}{RT} & \multicolumn{1}{c|}{0.99}  & \multicolumn{1}{c|}{7}  & \multicolumn{1}{c|}{7}  & \multicolumn{1}{c|}{7}  & \multicolumn{1}{c|}{7} & \multicolumn{1}{c|}{6} & 6             \\
                    & \multicolumn{1}{c|}{0.999} & \multicolumn{1}{c|}{10} & \multicolumn{1}{c|}{10} & \multicolumn{1}{c|}{10} & \multicolumn{1}{c|}{9} & \multicolumn{1}{c|}{9} & 6             \\ \midrule
\multirow{2}{*}{CT} & \multicolumn{1}{c|}{0.99}  & \multicolumn{1}{c|}{3}  & \multicolumn{1}{c|}{3}  & \multicolumn{1}{c|}{3}  & \multicolumn{1}{c|}{3} & \multicolumn{1}{c|}{3} & 3             \\
                    & \multicolumn{1}{c|}{0.999} & \multicolumn{1}{c|}{5}  & \multicolumn{1}{c|}{5}  & \multicolumn{1}{c|}{5}  & \multicolumn{1}{c|}{5} & \multicolumn{1}{c|}{5} & 5             \\ \bottomrule
\end{tabular}
\end{table}

\begin{table}[t]
\centering
\caption{HT and HT$^*$ optimal configuration for each reliability target and SF.}
\label{tab:ht_optimal}
\begin{tabular}{@{}ccccccc|cccc@{}}
\toprule
\multicolumn{1}{c|}{\multirow{2}{*}{\textbf{Scheme}}} & \multicolumn{2}{c|}{\textbf{Spreading Factor}}             & \multicolumn{4}{c|}{\textbf{SF7-SF11}} & \multicolumn{4}{c}{\textbf{SF12}} \\ \cmidrule(l){2-11} 
\multicolumn{1}{c|}{}                                 & \multicolumn{2}{c|}{\textbf{Parameter}}                    & $m$      & $n$     & $r$     & $M$     & $m$     & $n$    & $r$    & $M$    \\ \midrule
\multirow{2}{*}{\textbf{HT}}                          & \multirow{2}{*}{$\mathcal{T}$} & \multicolumn{1}{c|}{\textbf{0.99}}  & 2        & 1       & 3       & 5       & 2       & 1      & 3      & 5      \\
                                                      &                      & \multicolumn{1}{c|}{\textbf{0.999}} & 2        & 1       & 4       & 6       & 2       & 1      & 3      & 5      \\ \midrule
\multirow{2}{*}{\textbf{HT$^*$}}                      & \multirow{2}{*}{$\mathcal{T}$} & \multicolumn{1}{c|}{\textbf{0.99}}  & 1        & 1       & 2       & 3       & 1       & 1      & 2      & 3      \\
                                                      &                      & \multicolumn{1}{c|}{\textbf{0.999}} & 2        & 1       & 3       & 5       & 2       & 1      & 3      & 5      \\ \bottomrule
\end{tabular}
\end{table}

\section{Final Comments}
\label{sec:conclusion}
This paper proposed HT, a coded replication scheme suitable for \lorawan{}, which is a generalization of RT and CT.
We presented a detailed analysis in terms of outage probability and the number of devices, considering both connection and collision probabilities, showing that RT and CT never outperform HT.
The superiority of HT increases with the reliability target, making it suitable for industrial applications.
Moreover, HT has a larger set of configuration parameters than RT and CT, being more flexible and becoming able to better adapt to different deployments and requirements.
As a future work, we plan to test the proposed scheme in a testbed.

\bibliographystyle{IEEEtran}
\bibliography{references}

\end{document}